\documentclass[12pt,english]{iopart}

\usepackage[usenames,dvipsnames]{color}
\usepackage[normalem]{ulem}
\usepackage{graphicx}
\expandafter\let\csname equation*\endcsname\relax
\expandafter\let\csname endequation*\endcsname\relax
\usepackage{iopams}
\usepackage{amsmath}
\usepackage{amssymb,amsthm}
\usepackage{amsfonts}
\usepackage[breaklinks]{hyperref}
\usepackage[utf8]{inputenc}
\usepackage{color}
\usepackage{graphics}
\usepackage{epsfig}
\usepackage{bbm}
\usepackage{pstricks}
\usepackage{subfigure}
\usepackage{bm}
\usepackage{bbold}
\usepackage{cite}

\usepackage{diagbox}
\usepackage{cleveref}
\newcommand{\reffig}[1]{figure \ref{#1}}
\newcommand{\reffact}[1]{fact \ref{#1}}

\newcommand{\bra}[1]{\langle #1|}
\newcommand{\ket}[1]{|#1\rangle}

\newcommand{\trace}{\operatorname{Tr}}

\newtheorem{theorem}{Theorem}[section]
\newtheorem{fact}{Fact}

\newtheorem{corollary}{Corollary}[section]
\newtheorem{proposition}{Proposition}[section]
\newtheorem{definition}{Definition}

\begin{document}

\title{Joint separable numerical range and bipartite entanglement witness}

\author{Pan Wu$^1$ and Runhua Tang$^2$}%\thanks{616784940@qq.com}

\address{$^1$School of Data and Computer Science, Sun Yat-sen University, Guangzhou 510000, China}
\address{$^2$Network and Information Center, Guangdong Food and Drug Vocational College,Guangzhou 510520, China}
\eads{\mailto{wupanchina@163.com}, \mailto{magic8515@163.com}}

\begin{abstract}
The entanglement witness is an important tool to detect entanglement. In 2017 an idea considering a pair of Hermitian operators of product form was published, which is called ultrafine entanglement witnessing. In 2018 some rigorous results were given. Here we improve their work. First we point this idea can be directly derived from an earlier concept named joint separable numerical range and explain how it works as a series of witnesses. Second by a simple method we present a sufficient condition for an effective pair. Finally we prove this condition is necessary for optimization. [M.Gachechiladze \emph{et al. {\rm 2018} J. Phys. {\rm A}: Math.Theor.} {\bf 51} 36.]
\bigskip

\end{abstract}

%\maketitle
%%%%%%%%%%%%%%%%%%%%%%%%%%%%%%%%%%%%%%%%%%%%%
\section{Introduction}
%%%%%%%%%%%%%%%%%%%%%%%%%%%%%%%%%%%%%%%%%%%%%
The entanglement witness(in short, witness) is practical to judge whether a quantum state is entangled, which was proved NP-hard\cite{nphard}. A witness is defined as an Hermitian operator $W$ such that (\romannumeral1) $\trace(W\sigma) \geq 0$ for every  $\sigma\in S_{sep}$ and (\romannumeral2) $\trace(W\rho) < 0$ for at least one $\rho$, where $S_{sep}$ denotes the set of separable states. If $W$ satisfies (\romannumeral1) it is called \emph{block-positive}\cite{superlong}; (\romannumeral2) means $W$ is not positive semi-definite. Thus a negative expectation value of $W$ measuring $\rho$ establishes $\rho$ is entangled.

A witness can be constructed as a block-positive operator:
\begin{equation}
W_{min}(H)\overset{\text{def}}{=}H-\lambda^{\otimes}_{min}\mathbbm{1}\nonumber
\end{equation}
where $\lambda^{\otimes}_{min}$ denotes the minimum expectation value of Hermitian operator $H$ within $S_{sep}$. Similarly we can define $W_{max}(H)$. When the minimum eigenvalue of $H$ is less than $\lambda^{\otimes}_{min}$, $W_{min}(H)$ is really a witness. This kind of witness is called weakly optimal\cite{superlong}, which is no longer block-positive when subtracted by a positive operator. Furthermore if the set of entangled states detected by $W$ is not included by that of any other witness, we say $W$ is optimal. Authors of \cite{2000} proved that witness $W$ is optimal if and only if it is no longer block-positive when subtracted by a positive semi-definite operator.

Recently\cite{ultrafine} published an idea considering a pair of Hermitian operators $(H_1,H_2)$ to judge entanglement. To make it convenient for local measurement, they focus on the scenario when $H_i=A_{i}\otimes B_{i}$ where $A_i$ and $B_i$ are Hermitian for $i=1,2$. That is what we call ``product form" in Abstract. What we extract from\cite{ultrafine} is that we establish $\rho$ is entangled if the expectation value pair of it can not be that of any separable state. The set of pairs of expectation values of $(H_1,H_2)$ within $S_{sep}$ is the joint separable numerical range of $(H_1,H_2)$, which is derived from the concept joint numerical range.

The joint numerical range of $(H_1,H_2)$ is equivalent to the classic concept numerical range. Derived concepts of it are widely used in quantum theory(see\cite{chenjianxin} and Section 1.B of\cite{versatile}).  When the operators are of product form, the joint separable numerical range is the convex hull of a kind of product of two joint numerical ranges. When we apply this concept to judge entanglement, the essence is that although $S_{sep}$ is very hard to characterize, one can try to characterize its image in a low-dimensional space, i.e. the joint separable numerical range. For one Hermitian operator $H$, this image is $[\lambda^{\otimes}_{min},\lambda^{\otimes}_{max}]$; for a pair $(H_1,H_2)$, it is a more complex set in $\mathbb{R}^{2}$.

Let us call the above notion macro-view. Then the micro-view is that $\rho$ can be established to be entangled by $(H_1,H_2)$ if and only if it can be witnessed by at least one $W_{min}(k_1H_1+k_2H_2)$, which is provided by\cite{German}. This is similar to the fact that every entangled state can be detected by a witness. Hence we reckon generally a pair is finer than a single operator since it generally represents a series of weakly optimal witnesses. Thus a basic problem arises that \emph{when at least one $W_{min}(k_1H_1+k_2H_2)$ is really a witness}, i.e., $(H_1,H_2)$ can detect at least one entangled state (in short, $(H_1,H_2)$ is \emph{effective}). Authors of \cite{German} found a canonical necessary condition: $A_1A_2\neq A_2A_1$ and $B_1B_2\neq B_2B_1$. Moreover, this is sufficient in $\mathbb{C}^{2}\otimes \mathbb{C}^{2}$ system. They also discussed the $\mathbb{C}^{2}\otimes \mathbb{C}^{3}$ scenario.

However, the method in \cite{German} based on the Perturbation Theory is complicated and its result needs extension to higher dimensional spaces since PPT criteria\cite{PPT}. By a simple method we derive a sufficient condition for effective $(H_1,H_2)$ \emph{independent of dimension}. This method is based on the trivial fact that the sum of two product vectors is generally entangled. Next by orthogonally dividing the whole space into invariant subspaces, we provide a more powerful conclusion: this sufficient condition is \emph{necessary for optimal} $(H_{1},H_{2})$. This method also gives instruction in Section 4.

This paper is organized as follows: In Section 2 we supplement Section 1 with the rigorous definition and corresponding explanation. Section 3 is the solution to the basic problem mentioned in the previous paragraph. Section 4 is a short conclusion and the plan for further research. Some explanations for the end of Section 3 and the beginning of Section 4 are in Appendix A and Appendix B respectively.

\section{Definitions and more explanation}
To avoid overwhelming readers with many definitions in this section, we stress that only some of them are necessary for understanding Section 3, especially the labeled equations. The others are assistant. For example, some will be discussed in Section 4.

Let $d$ be the dimension of a system. Let $d_A$ and $d_B$ be the dimensions of subsystems respectively. Let $M$ act on $\mathbb{C}^{d}$ or $\mathbb{C}^{d_A}\otimes \mathbb{C}^{d_B}$. From now on we default that $\alpha\in\mathbb{C}^{d_A}$, $\beta\in\mathbb{C}^{d_B}$ and so do the normalized vectors $\ket{\alpha}$,$\ket{\beta}$. Define
\begin{align}
&\Lambda(M)\overset{\text{def}}{=}\{\bra{\phi}M\ket{\phi}|\ket{\phi}\in\mathbb{C}^{d}\},\nonumber\\
&\Lambda^{\otimes}(M)\overset{\text{def}}{=}\{\bra{\alpha\beta}M\ket{\alpha\beta}\},\nonumber\\
&\Lambda^{sep}(M)\overset{\text{def}}{=}\{\trace(M\rho)|\rho \in S_{sep}\},\nonumber
\end{align}
where $\Lambda(M)$ denotes the numerical range of $M$. $\Lambda(M)$ is convex and when $d=2$ it is generally a ellipse disc\cite{100nian}.

Let $H_1=(M^{*}+M)/2$ and $H_2=i(M^{*}-M)/2$. It is simple to verify that both $H_{1}$ and $H_{2}$ are Hermitian and  $M=H_{1}+iH_{2}$. Thus represented on $\mathbb{R}^{2}$, $\Lambda(M)$ is equivalent to $\{(\bra{\phi}H_{1}\ket{\phi},\bra{\phi}H_{2}\ket{\phi})|\ket{\phi}\in\mathbb{C}^{d}\}$, which is named joint numerical range of $(H_{1},H_{2})$ and denoted by $\Lambda(H_{1},H_{2})$. Similarly when $H_{1}$ and $H_{2}$ act on $\mathbb{C}^{d_{A}}\otimes\mathbb{C}^{d_{B}}$,
\begin{align}
&\Lambda^{\otimes}(H_{1},H_{2})\overset{\text{def}}{=}\{(\bra{\alpha\beta}H_{1}\ket{\alpha\beta},\bra{\alpha\beta}H_{2}\ket{\alpha\beta}) \}, \nonumber\\
&\Lambda^{sep}(H_{1},H_{2})\overset{\text{def}}{=}\{(\trace(H_{1}\sigma),\trace(H_{2}\sigma))|\sigma\in S_{sep}\} \nonumber
\end{align}
where $\Lambda^{sep}(H_{1},H_{2})$ denotes the joint separable numerical range of $(H_{1},H_{2})$. Since the definition of $S_{sep}$, clearly $\Lambda^{sep}(H_{1},H_{2})$ is the convex hull of $\Lambda^{\otimes}(H_{1},H_{2})$. $\Lambda^\otimes(H_1,H_2)$ may not be convex: consider $H_{1}=\ket{00}\bra{00}$ and $H_{2}=\ket{11}\bra{11}$\cite{withtensor}. This fact corresponds to the correction provided by\cite{German} against theorem 1 of\cite{ultrafine}. The above definitions can be seen in\cite{versatile} or\cite{chenjianxin}.

We say $\lambda^{\otimes}_{min}(\lambda^{\otimes}_{max})$ corresponds to Hermitian operator $H$ if it is the minimum(maximum) value of $\{\bra{\alpha\beta}H\ket{\alpha\beta}\}$. Since $\{\ket{\alpha\beta}\}$ is a close set we can obtain $\lambda^{\otimes}_{min}$ and $\lambda^{\otimes}_{max}$. We say $H$ is block-positive if $\lambda^{\otimes}_{min}\geq0$ where $\lambda^{\otimes}_{min}$ corresponds to $H$. Then as the case with one parameter, $\Lambda^{sep}(H)=[\lambda^{\otimes}_{min},\lambda^{\otimes}_{max}]$. These have been stated in Section 1.

According to the 3rd paragraph of Section 1 $\rho$ is established to be entangled by $(H_{1},H_{2})$ iff
\begin{equation}\label{zhunze}
(\trace(H_1\rho),\trace(H_2\rho))\notin \Lambda^{sep}(H_{1},H_{2}).
\end{equation}
Since $\Lambda(H_{1},H_{2})$ is convex then $(H_{1},H_{2})$ is \emph{not} effective to detect entanglement iff
\begin{equation}\label{effective}
\Lambda^{sep}(H_{1},H_{2})=\Lambda(H_{1},H_{2}).
\end{equation}
Generally \eqref{effective} does not hold. One example where \eqref{effective} holds is Observation 2 of\cite{chenjianxin}. Another example is when $H_i=A_{i}\otimes B_{i}$ where $A_i$ and $B_i$ are Hermitian for $i=1,2$, if $A_1A_2=A_2A_1$ or $B_1B_2=B_2B_1$ then \eqref{effective} holds, which is cited in Section 1. For complex matrices there is an earlier and similar property cited in\cite{withtensor} that if $M_A$ is normal then $\Lambda(M_A\otimes M_B)=\Lambda^{sep}(M_A\otimes M_B)$.

We say $H_1\geq H_2$ if $H_1-H_2$ is positive semi-definite. From now on we \emph{default $W$ to be a block-positive operator}. Define the detection range of it like\cite{2000}:
\begin{equation}
D(W)\overset{\text{def}}{=}\{\rho|\trace(W\rho)<0,\rho\geq0, \trace(\rho)=1\}.
\end{equation}

Then we specifically explain why \eqref{zhunze} is equivalent to the statement that at least one $W_{min}(k_1H_1+k_2H_2)$ can witness $\rho$, which is very briefly explained by\cite{German} behind its (2). From Separating Hyperplane Theorem, \eqref{zhunze} holds iff there exist $k_{1},k_{2}\in \mathbb{R}$ such that
\begin{equation}
k_{1}\trace(H_{1}\rho)+k_{2}\trace(H_{2}\rho)<k_{1}\trace(H_{1}\sigma)+k_{2}\trace(H_{2}\sigma)\nonumber
\end{equation}
holds for every $\sigma\in S_{sep}$, which means for any point out of $\Lambda^{sep}(H_{1},H_{2})$ there exists a line separating it. Hence \eqref{zhunze} holds iff there exist $k_{1},k_{2}\in \mathbb{R}$ such that
\begin{equation}
k_{1}\trace(H_{1}\rho)+k_{2}\trace(H_{2}\rho)<\lambda^{\otimes}_{min} \nonumber
\end{equation}
where $\lambda^{\otimes}_{min}$ corresponds to $k_{1}H_{1}+k_{2}H_{2}$. That means at least one $W_{min}(k_{1}H_{1}+k_{2}H_{2})$ can witness $\rho$. Thus the set of entangled states detected by $(H_{1},H_{2})$ can be denoted as follows:
\begin{equation}\label{union}
\widetilde{D}(H_1,H_2)\overset{\text{def}}{=}\bigcup\limits_{k_{1},k_{2}\in \mathbb{R}}D(W_{min}(k_{1}H_{1}+k_{2}H_{2})).
\end{equation}

An example is when $H_{1}=X\otimes X$ and $H_{2}=Z\otimes Z$($X,Z$ are Pauli operators acting on $\mathbb{C}^2$), $\Lambda(X,Z)$ is the circular disc $x^2+y^2\leq1$(see Example 1 of \cite{arx}); $\Lambda(H_{1},H_{2})$ is the square $|x|\leq1$ and $|y|\leq1$; $\Lambda^{\otimes}(H_{1},H_{2})$=$\Lambda^{sep}(H_{1},H_{2})$ is the square $|x|+|y|\leq 1$, which is depicted in figure 1 of\cite{German}. In \reffig{fig1} of our paper we illustrate that generally a tangent line to $\Lambda^{sep}(H_{1},H_{2})$  represents a witness. By the way, since the vertex (1,1) of $\Lambda(H_{1},H_{2})$ is obtained on $\ket{\phi}=(\ket{00}+\ket{11})/\sqrt{2}$ then $\bra{\phi}(X\otimes X+Z\otimes Z)\ket{\phi}=2$, which is equivalent to the fact that $\ket{\phi}$ provides the maximal violation to CHSH inequality in 2.6 of\cite{book}.
\begin{figure}[h]
    \centering
    \includegraphics[width=0.6\columnwidth]{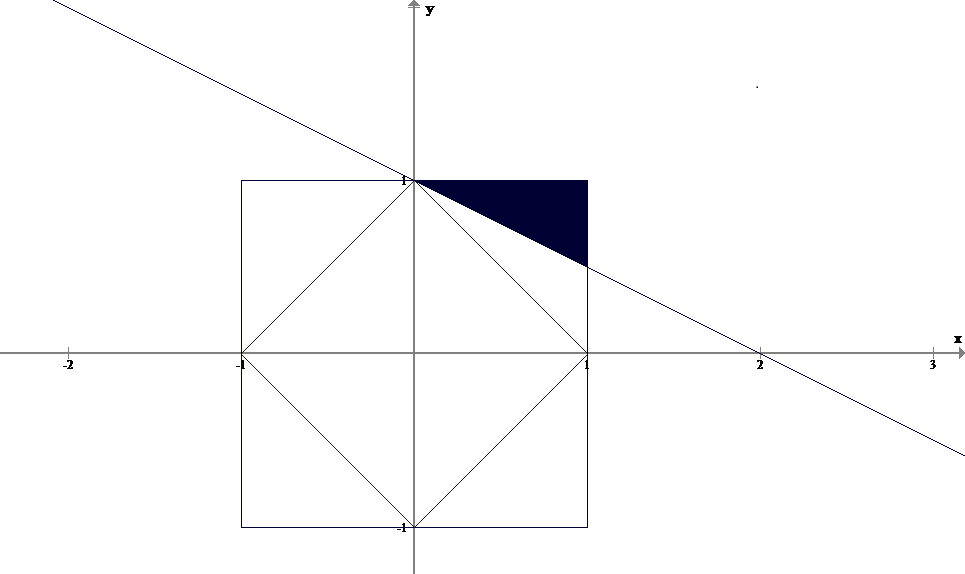}
    \caption{When $H_{1}=X\otimes X$ and $H_{2}=Z\otimes Z$, $\Lambda(H_{1},H_{2})$ and $\Lambda^{sep}(H_{1},H_{2})$ is the outer and inner square respectively. A block-positive operator $W_{min}(k_{1}H_{1}+k_{2}H_{2})$ is generally a witness represented by the line $k_{1}x+k_{2}y=\lambda^{\otimes}_{min}$ tangent to $\Lambda^{sep}(H_{1},H_{2})$(the inner square). If $k_{1}=-1$ and $k_{2}=-2$ this tangent line is $x+2y=2$ and the entangled states detected by corresponding witness is represented by coloured part. }
    \label{fig1}
\end{figure}
%%%%%%%%%%%%%%%
\section{When is a product pair effective}
%%%%%%%%%%%%%%%
As is explained before, a pair of Hermitian operators $(H_{1},H_{2})$ is effective iff there exist $k_1,k_2\in\mathbb{R}$ such that $W_{min}(k_1H_1+k_2H_2)$ is a witness. \emph{From now on we default that $H_i=A_i\otimes B_i$ where $A_i$ and $B_i$ are Hermitian for $i=1,2$}. Then we derive that if the common eigenvectors of $A_1$ and $A_2$ only correspond to eigenvalue 0 and that also holds for $B_1$ and $B_2$, then $(H_{1},H_{2})$ can detect entanglement, which was called the sufficient condition before. Moreover, this scenario includes all the optimal $(H_1,H_2)$. Hence we claim this sufficient condition is almost necessary.
\begin{fact}
\label{entangle}
If $\alpha_1\otimes \beta_1+\alpha_2\otimes \beta_2$ is a product vector, then $\{\alpha_1$, $\alpha_2\}$ is linear dependent or $\{\beta_1$, $\beta_2\}$ is linear dependent.
\end{fact}

\begin{fact}
\label{trivial}
If $\alpha_1\otimes \beta_1\neq0$ and $\alpha_2\otimes \beta_2\neq0$ meanwhile $\{\alpha_1\otimes \beta_1$, $\alpha_2\otimes \beta_2\}$ is linear dependent, then $\{\alpha_1$,$\alpha_2\}$ is linear dependent; so is $\{\beta_1$,$\beta_2\}$.
\end{fact}

\begin{proposition}
\label{initial}
If $k_1,k_2\in \mathbb{R}$, $k_1,k_2\neq0$ and $\lambda\neq0$ is an eigenvalue of
\begin{equation}\label{combination}
k_1A_1\otimes B_1+k_2A_2\otimes B_2,
\end{equation}
then for any product eigenvector $\alpha\otimes \beta$ corresponding to $\lambda$, $\alpha$ is a common eigenvector of $A_1$ and  $A_2$ or $\beta$ is a common eigenvector of $B_1$ and  $B_2$.
\end{proposition}

\begin{proof}
Initially we have
\begin{equation}
k_1A_1\alpha\otimes B_1\beta+k_2A_2\alpha\otimes B_2\beta=\lambda\alpha\otimes \beta.\nonumber
\end{equation}
From \reffact{entangle} and without loss of generosity, we suppose $\{k_1A_1\alpha$,$k_2A_2\alpha\}$ is linear dependent. Since $k_1\neq0$ and $k_2\neq0$ then $\{A_1\alpha$,$A_2\alpha\}$ is linear dependent. Without loss of generosity let $A_2\alpha$ =$kA_1\alpha$, then
\begin{equation}
A_1\alpha\otimes (k_1B_1\beta+ kk_2B_2\beta)=\lambda\alpha\otimes \beta.\nonumber
\end{equation}
Since $\lambda\neq0$ and \reffact{trivial}, we establish $\{\alpha$, $A_1\alpha\}$ is linear dependent, which implies there exists $\eta$ such that $A_1\alpha=\eta\alpha$. Then $A_2\alpha=\eta k\alpha$. Thus $\alpha$ is a common eigenvector of $A_1$ and $A_2$.
\end{proof}

\begin{theorem}
\label{kernel}
If $A_1$, $A_2$ do not have any common eigenvector and neither do $B_1$, $B_2$, then for any $k_1,k_2\in\mathbb{R}$ such that $k_1k_2\neq0$, the eigenvectors corresponding to the minimum eigenvalue of \eqref{combination} are entangled; otherwise the eigenvectors corresponding to the maximum eigenvalue of \eqref{combination} are entangled.
\end{theorem}
\noindent\textbf{Remark. }In physics context, the ground states and the most excited states include the eigenvectors corresponding to the minimum and maximum eigenvalue respectively.
\begin{proof}
From the condition part, it is clear that the operator set $\{A_1$,$A_2\}$ is linear independent. So is $\{B_1$,$B_2\}$. Extending from \reffact{trivial}, $\{A_1\otimes B_1,A_2\otimes B_2\}$ is linear independent. Then for any $k_1,k_2\neq0$, \eqref{combination}$\neq0$, which implies not both the minimum and maximum eigenvalue of \eqref{combination} are zero. Without loss of generosity suppose the minimum eigenvalue of \eqref{combination} is not zero. Then from the condition part and \cref{initial}, any product vector can not be an eigenvector of \eqref{combination} corresponding to the minimum eigenvalue.
\end{proof}

From \cref{kernel} if the condition in \cref{kernel} holds, then for any $k_1,k_2\neq0$, $W_{min}(k_1H_1+k_2H_2)$ or $W_{max}(k_1H_1+k_2H_2)$ is a witness. As is mentioned in Section 1, proposition 1 of\cite{German} proved that when $d_A=d_B=2$, if $A_1A_2\neq A_2A_1$ and $B_1B_2\neq B_2B_1$ then $(H_1,H_2)$ is effective. Since when $d_A=d_B=2$, $A_1A_2\neq A_2A_1\Leftrightarrow$$A_1$ and $A_2$ can not be orthogonally simultaneously diagonalized $\Leftrightarrow A_1$ and $A_2$ do not have any common eigenvector, \cref{kernel} of this paper implies it. Moreover, since $W_{min}(H)=W_{max}(-H)$ then \cref{kernel} means at least ``half" of $W_{min}(k_1H_1+k_2H_2)$ are witnesses, while\cite{German} only guarantees a little part since the idea of limit.

In retrospect to the proof of \cref{initial}, it is clear that the common eigenvector $\alpha$ should not only correspond to eigenvalue 0 since $\lambda\neq0$. Since the proof of \cref{kernel} still generally applies, we can make \cref{kernel} milder as follows:
\begin{corollary}\label{only0common}
Let $H_1\neq0$ or $H_2\neq0$. If the common eigenvectors of $A_1$ and $A_2$ only correspond to eigenvalue 0 and that also holds for $B_1$ and $B_2$, then $W_{min}(k_1H_1+k_2H_2)$ or $W_{max}(k_1H_1+k_2H_2)$ is a witness for any $k_1,k_2\in\mathbb{R}$ such that $k_1k_2\neq0$.
\end{corollary}

An example of \cref{only0common} is $H_i=\ket{\alpha_i}\bra{\alpha_i}\otimes \ket{\alpha_i}\bra{\alpha_i}$ for $i=1,2$ where $|\langle\alpha_1|\alpha_2\rangle|^{2}$ is neither 0 nor 1. Let us explain this example in another way: Consider $\ket{\alpha_1}\bra{\alpha_1}$ and $\ket{\alpha_2}\bra{\alpha_2}$ as operators acting on $Q=$span$\{\ket{\alpha_1},\ket{\alpha_2}\}$. Since $|\langle\alpha_1|\alpha_2\rangle|^{2}$ is neither 0 nor 1, then $\ket{\alpha_1}\bra{\alpha_1}$ and $\ket{\alpha_2}\bra{\alpha_2}$ do not have any common eigenvector. From \cref{kernel} when $k_1k_2\neq0$, all the pure ground states or all the pure most excited states of \eqref{combination} are entangled vectors in $Q\otimes Q$.

Naturally an intuition comes that common eigenvectors not permitted by the condition part of \cref{only0common} are redundant. The next theorem reveals they really are. That is the main reason why we claim the sufficient condition in \cref{only0common} is almost necessary. To prove this we need the following two facts and one definition.
\begin{fact}\label{finer}
$W\geq W'\Rightarrow D(W)\subseteq D(W')$.
\end{fact}
\begin{definition}\label{projector}
Let $Q$ be an invariant subspace acted by $H$. Linear operator $H'$ such that $H'\phi=H\phi$ for $\phi\in Q$ and $H'\phi=0$ for $\phi\in Q^{\perp}$ is called the projector of $H$ onto $Q$, which is denoted by $H_Q$.
\end{definition}
\begin{fact}\label{fenkuai}
$Q_i=Q_i^{A}\otimes Q_i^{B}$ where $Q_i^{A}\subseteq\mathbb{C}^{d_A}$ and $Q_i^{B}\subseteq \mathbb{C}^{d_B}$ for $i=1,...,r$. If $Q_1,...,Q_r$ are orthogonal to each other and the sum of them equals to $\mathbb{C}^{d_A}\otimes \mathbb{C}^{d_B}$, then $W_{Q_1},...,W_{Q_r}$ are block-positive.
\end{fact}
\noindent\textbf{Remark. }To be concrete, one can consider $W_{Q_1},...,W_{Q_r}$ as block diagonal matrices.
\begin{theorem} \label{xinxue}
Let the common eigenvector $\ket{\alpha}$ of $A_1$ and $A_2$ correspond to eigenvalues $a_1$ and $a_2$ respectively. Let $A_i'=A_i-a_i\ket{\alpha}\bra{\alpha}$ for $i=1,2$. Then $\widetilde{D}(H_1,H_2)\subseteq\widetilde{D}(A_1'\otimes B_1,A_2'\otimes B_2)$.
\end{theorem}
\begin{proof}
Let $\lambda^{\otimes}_{min}$ correspond to $k_1H_1+k_2H_2$. Let $Q_1=$span$\{\ket{\alpha}\}\otimes \mathbb{C}^{d_B}$ and $Q_2=$span$\{\ket{\alpha}\}^{\perp}\otimes \mathbb{C}^{d_B}$. Let
\begin{align}
&W_1=\sum\limits_{i=1}^{2}(k_ia_i\ket{\alpha}\bra{\alpha}\otimes B_i)-\lambda_{min}^{\otimes}\mathbbm{1}_{Q_1}, \nonumber\\
&W_2=\sum\limits_{i=1}^{2}(k_iA_i'\otimes B_i)-\lambda_{min}^{\otimes}\mathbbm{1}_{Q_2}.\nonumber
\end{align}
Then $W_{min}(k_1H_1+k_2H_2)=W_1+W_2$. Clearly $Q_1\perp Q_2$, meanwhile $W_1$ and $W_2$ are projectors of $W_{min}(k_1H_1+k_2H_2)$ onto $Q_1$ and $Q_2$ respectively. From \reffact{fenkuai}, $W_1$ and $W_2$ are block-positive. For any $\phi\in\mathbb{C}^{d_A}\otimes \mathbb{C}^{d_B}$ let $\phi=\phi_1+\phi_2$ where $\phi_1\in Q_1$ and $\phi_2\in Q_2$. Since span$\{\ket{\alpha}\}$ is 1-dimensional then $\phi_1$ must be a product vector. Then $\phi^{*}W_1\phi=\phi_1^{*}W_1\phi_1\geq0$, which means $W_1\geq0$.

Thus $W_{min}(k_1H_1+k_2H_2)\geq W_2$. Since $W_2$ is block-positive then $W_2\geq W_{min}(\sum\limits_{i=1}^{2}(k_iA_i'\otimes B_i))$, which means
\begin{equation}
W_{min}(k_1H_1+k_2H_2)\geq W_{min}(\sum\limits_{i=1}^{2}(k_iA_i'\otimes B_i)). \nonumber
\end{equation}
Finally from \reffact{finer} and \eqref{union} the conclusion is clear.
\end{proof}
From \cref{xinxue} if we subtract projectors of $A_i$ or $B_i$ onto common eigenspaces, the new pair will be finer. Moreover, \cref{only0common} guarantees the new pair to be effective. Since there is no sign that the existence of common eigenvectors not only corresponding to eigenvalue 0 can simplify computation(we will explain this opinion in Appendix A), then we say the sufficient condition in \cref{only0common} is almost necessary for practice.
%%%%%%%%%%
\section{Conclusion and prospect}
%%%%%%%%%%

So far we have very clearly described the basic problem in Section 1 and improved the solution to a large extent. During the past time we have derived other sufficient conditions for an effective pair, which are milder in a sense. However, from \cref{xinxue} and Appendix A, we now think the milder parts of them are not valuable. One of these sufficient conditions  will be stated in Appendix B for interested readers.

For further research, we shall design $(H_1,H_2)$ for $\mathbb{C}^{2}\otimes \mathbb{C}^{4}$ system. One principle is that the computation of $\Lambda^{sep}(H_1,H_2)$ should be practical. Another principle is the supplement to PPT criteria. Any block-positive operator acting on $\mathbb{C}^{2}\otimes \mathbb{C}^{2}$ system is decomposable(see Section 3.C of\cite{2000}), which means it can not detect any PPT entanglement. Hence similarly by the proof of \cref{xinxue}, we assert there should not exist an orthogonal decomposition of $\mathbb{C}^{4}$ such that the decomposed two subspaces are invariant acted by $B_1$ and $B_2$ respectively, which means $B_1$ and $B_2$ should not be ``reducible"\cite{3X3}.

From Section 2 we know $\Lambda^{\otimes}(H_1,H_2)$ is the set $\{(x_1x_2,y_1y_2)\}$ where $(x_1,y_1)\in\Lambda(A_1,A_2)$ and $(x_2,y_2)\in\Lambda(B_1,B_2)$. That corresponds to ``a kind of product" in the 4th paragraph of Section 1. Similarly when $M_A$ and $M_B$ are complex matrices, $\Lambda^{\otimes}(M_A\otimes M_B)$ is the Minkowski product of $\Lambda(M_A)$ and $\Lambda(M_B)$\cite{withtensor}. However, we have not found a profound relationship between the two kinds of product. Otherwise we can take advantage of the results related to Minkowski product for computation.

For the computation in $\mathbb{C}^{2}\otimes \mathbb{C}^{4}$ and $\mathbb{C}^{3}\otimes \mathbb{C}^{3}$ systems, the papers\cite{3X3} introduced are useful. For the computation of the joint numerical range for the extended form $(H_1,H_2,H_3)$, \cite{3X3} itself may be crucial.
%%%%%%%%%
\ack
This work is supported by the National Natural Science Foundation of China under Grants No. 61672007. We thank the volunteer instructor BangHai Wang from Guangdong University of Technology. We are especially grateful to him for his introduction of\cite{German} and an important tip for \cref{kernel}. %We also thank the authors of\cite{German} for much inspiration.
\appendix
\section{}
%\appendix{hard computation from "actual" common eigenspace}
In this section the definition of symbols follows \cref{xinxue}. We will explain why we approximately think $\Lambda^{sep}(H_1,H_2)$ is generally more difficult to compute than $\Lambda^{sep}(A_1'\otimes B_1,A_2'\otimes B_2)$, where $A_1'$ and $A_2'$ are defined in \cref{xinxue}. Let us focus on $\Lambda(A_1,A_2)$.

Let the point $P_0$ be $(a_1,a_2)$. From \fref{fig1}, we reckon a $\Lambda(A_1,A_2)$ centring around $(0,0)$ or at least including $(0,0)$ can simplify computation. Suppose $\Lambda(A_1,A_2)$ includes $(0,0)$. Then we assert
\begin{equation}\label{conv}
\Lambda(A_1,A_2)=\rm{conv}(\Lambda(A_1',A_2')\cup P_0).
\end{equation}
\setcounter{section}{1}
where conv$(S)$ means the convex hull of $S$. From the premise that $\Lambda(A_1,A_2)$ includes $(0,0)$, the fact that $\Lambda(A_1,A_2)$ is convex and the method that for $\psi\in\mathbb{C}^{d_A}$ we divide it as $\psi=\alpha+\psi'$ where $\alpha\perp\psi'$, we can prove \eqref{conv}.

From \eqref{conv} if we postulate $\Lambda(A_1,A_2)$ includes $(0,0)$, then $\Lambda(A_1,A_2)$ includes $\Lambda(A_1',A_2')$. Hence $\Lambda^{sep}(H_1,H_2)$ includes $\Lambda^{sep}(A_1'\otimes B_1,A_2'\otimes B_2)$. Generally the included set is easier to compute. Moreover, if we consider more common eigenvectors which correspond to more points like $P_0$, the $P_0$ in \eqref{conv} will be substituted by the polygon spanned by those points. This polygon probably makes $\Lambda(A_1,A_2)$ more complex but not more symmetric.
\section{}
\begin{proposition}
Let the ground states of $H_1$ be non-degenerate. Statements {\rm(a)} and {\rm(b)} are as follows: {\rm(a)} For the ground state $\ket{\alpha \beta}$ of $H_1$, $\ket{\alpha}$ is an eigenvector of $A_2$ or $\ket{\beta}$ is an eigenvector of $B_2$. {\rm(b)} There exists a $\delta>0$ such that if $|x|<\delta$ then at least one ground state of $H_1+xH_2$ is separable. Then {\rm(a)}$\Leftrightarrow${\rm(b)}.
\end{proposition}

The proof is mainly based on \cref{initial} and the following fact:
\begin{equation}
\lim\limits_{x\to 0}E_1(H_1+xH_2)=E_1(H_1)  \nonumber
\end{equation}
\setcounter{section}{1}
where $x\in\mathbb{R}$ and $E_1(H)$ denotes the minimum eigenvalue of $H$. From the direction (b)$\Rightarrow$(a) we assert in non-degenerate scenario, $\neg$(a) guarantees an effective pair and is a milder sufficient condition than the condition in \cref{kernel}. From the other direction (a)$\Rightarrow$(b) we assert in non-degenerate scenario, $\neg$(a) is milder than any other sufficient condition derived from the Perturbation Theory that \cite{German} used because this theory is based on the idea of limit\cite{perturbation}.

\end{document}